\documentclass[letterpaper]{article}
\usepackage{aaai}
\usepackage{times}
\usepackage{helvet}
\usepackage{courier}
\nocopyright
\newcommand{\ourtitle}{Elicitation for Aggregation}
\usepackage{color}
\usepackage{amsmath,amssymb,amsfonts,amsthm,latexsym}
\usepackage{bbm} 
\usepackage{url, algorithm, algorithmic, verbatim}
\usepackage{etoolbox,xparse,xspace}

\newtheorem{theorem}{Theorem}

\newtheorem{definition}{Definition}
\newtheorem{lemma}[theorem]{Lemma}
\newtheorem{corollary}[theorem]{Corollary}

\newtheorem{conjecture}{Conjecture}

\newcommand{\Comments}{1}
\newcommand{\mynote}[2]{\ifnum\Comments=1\textcolor{#1}{#2}\fi}
\newcommand{\ian}[1]{\mynote{blue}{[IAK: #1]}}
\newcommand{\raf}[1]{\mynote{green}{[RMF: #1]}}
\newcommand{\yc}[1]{\mynote{red}{[YC: #1]}}

\newcommand{\tr}{\top}

\newcommand{\E}{\mathop{\mathbb{E}}}
\newcommand{\Var}{\mathop{\mathrm{Var}}}

\newcommand{\N}{\mathrm{N}}

\newcommand{\R}{\mathcal{R}}

\newcommand{\X}{\mathcal{X}}

\def\reals{\mathbb{R}}

\newcommand{\argmax}{\mathop{\mathrm{argmax}}}

\newcommand{\inprod}[1]{\left\langle #1 \right\rangle}

\newcommand{\defeq}{\doteq}

\newcommand{\ones}{\mathbbm{1}}

\renewcommand{\Comments}{0}

\begin{document}
\title{\ourtitle}
\author{Rafael M. Frongillo\\Harvard University\\\url{raf@cs.berkeley.edu} \And Yiling Chen\\Harvard University\\\url{yiling@seas.harvard.edu} \And Ian A. Kash\\Microsoft Research\\\url{iankash@microsoft.com}}
\maketitle
\begin{abstract}\vspace{-2pt}
We study the problem of eliciting and aggregating probabilistic information from multiple agents.
In order to successfully aggregate the predictions of agents, the principal needs to elicit some notion of \emph{confidence} from agents, capturing how much experience or knowledge led to their predictions.
To formalize this, we consider a principal who wishes to elicit predictions about a random variable from a group of Bayesian agents, each of whom have privately observed some independent samples of the random variable, and hopes to aggregate the predictions as if she had directly observed the samples of all agents.
Leveraging techniques from Bayesian statistics, we represent confidence as the number of samples an agent has observed, which is quantified by a hyperparameter from a conjugate family of prior distributions.
This then allows us to show that if the principal has access to a few samples, she can achieve her aggregation goal by eliciting predictions from agents using proper scoring rules.
In particular, if she has access to one sample, she can successfully aggregate the agents' predictions if and only if every posterior predictive distribution corresponds to a unique value of the hyperparameter.
Furthermore, this uniqueness holds for many common distributions of interest.
When this uniqueness property does not hold, we construct a novel and intuitive mechanism where a principal with two samples can elicit and optimally aggregate the agents' predictions. 
\end{abstract}

\section{Introduction}

Imagine that a principal, Alice, wishes to estimate the probability of rain tomorrow.  She consults two agents, Bob who says 80\%, and Carol who says 10\%.  How should Alice aggregate these two widely disparate predictions?  If she knew that Bob happened to have spent the day studying radar imagery, whereas Carol just looked outside for a second, it would seem obvious that Alice should give much higher weight to Bob's prediction than Carol's.  In other words, in order to aggregate these predictions, Alice needs to know the agents' \emph{confidence} about their reports.

The aggregation of probabilistic information is an important problem in many domains, from multiagent systems to crowdsourcing.  In this paper, we propose a general method of eliciting predictions together with a measure of confidence about those predictions, and show how to use this information to optimally aggregate in many situations.

We consider a Bayesian model where a principal, who can consult a group of risk-neutral agents, wishes to obtain an informed prediction about a random variable. The random variable follows a parameterized distribution that is generated by some unknown parameters, the prior distribution of which is common knowledge. Each agent privately observes some independent samples of the random variable and forms a belief about it. The principal then elicits the agents' predictions of the random variable, and her goal is to optimally aggregate agents' private beliefs based on these predictions --- to compute the distribution of the random variable as if she had observed the samples of all agents.   

This paper focuses on designing elicitation mechanisms to achieve this optimal aggregation. We show that when the prior distribution of the unknown parameters comes from a {\em conjugate prior family} of the distribution of the random variable, the principal can {\em leverage a few independent samples} that she observes to successfully elicit enough information from the agents to achieve the optimal aggregation. This relies on important properties of the conjugate prior family.  
Intuitively,  we use the hyperparameter of a distribution in the conjugate family to quantify the confidence of an agent's belief as the hyperparameter encodes information about the samples that the agent has observed. Our mechanisms work by eliciting predictions that allow the principal to infer the confidence of the agents and then make use of the confidence to achieve the optimal aggregation.  

In particular, we prove that the principal can leverage a single sample to achieve optimal aggregation if and only if each distribution (modulo an equivalence relation) in the conjugate family maps to a unique hyperparameter. With this, we demonstrate how elicitation and optimal aggregation work for many common distributions of the random variable, including the Poisson, Normal, and uniform distributions, among others.

When the unique mapping condition is not satisfied, such as in the rain example above, we show that the hyperparameter of an agent's posterior distribution cannot be inferred with the principal's single sample.  Fortunately, in this setting we construct a mechanism where the principal can still achieve the optimal aggregation if she has access to \emph{two} independent samples of the random variable.  Our mechanism simply asks each agent for his believed distribution of the first sample, and the likelihood that the two samples are the same.  We show that this simple and intuitive approach gives the principal second-order information about agents' beliefs, which is enough to achieve optimal aggregation.

\subsection{Related Work}

Our problem simultaneously considers both one-shot elicitation of information from multiple agents and the subsequent aggregation of the information. 

In one-shot elicitation, the principal interacts with each agent independently and the agents report their predictions without knowing others' predictions. There is a rich literature on mechanisms for one-shot elicitation. The simplest is the classical proper scoring rules~\cite{Brier:50,Win:69,Savage:71,gneiting2007strictly}, which incentivize risk-neutral agents to honestly report their predictions. Proper scoring rules are the building blocks for most elicitation mechanisms, including our mechanisms in this paper. To reduce the total payment of the principal, researchers design shared scoring rules~\cite{Kilgour:04,Johnstone:2007} and wagering mechanisms~\cite{Lambert:08b,L+14,Chen:2014c} that have various desirable theoretical properties. Both shared scoring rules and wagering mechanisms engage agents in a one-shot betting to elicit their information and do not require the principal to subsidize the betting. In contrast to our problem, all these one-shot elicitation mechanisms do not consider the aggregation of the elicited information. 


Sequential mechanisms have been designed to both elicit and aggregate information from agents. Most well known probably are prediction markets~\cite{Ber:01,Wol:04}, especially the market scoring rules mechanism~\cite{Han:03,Han:07}, where agents can sequentially interact with the market mechanism for multiple times to reveal their information. Information aggregation happens when agents update their beliefs after observing other agents' activities in the market. However, the dynamic nature of these mechanisms can induce complicated strategic play and obfuscate individual-level information~\cite{hansen-manipulation,Chen:10,Gao:2013}. In this paper, we let the principal rather than the agents take the responsibility of aggregating information, and couple aggregation with one-shot elicitation that is incentive compatible for the agents.    

To achieve optimal aggregation, the principal in our paper needs to know the confidence of agents' predictions. The work of \citeauthor{fang2007putting} (\citeyear{fang2007putting})
is the closest to ours in this perspective. They consider the one-shot elicitation of both agents' predictions and the precision of their predictions and then use the elicited precision to optimally aggregate. They use Normal distributions to model both the distribution of the random variable and the prior distribution of the unknown parameters. We consider general parameterized distributions of the random variable and their corresponding conjugate priors, which include the model of \citeauthor{fang2007putting} (\citeyear{fang2007putting}) as a special case. 

\section{Model and Background}
\label{sec:conf-model}

We introduce our model, which describes how agents form their beliefs, the principal's elicitation mechanism, the principal's aggregation goal, and a family of parameterized prior distributions that we will focus on in this paper. 

\subsection{Beliefs of Agents}
The principal would like to get information from $m$ agents about a random variable with observable outcome space $\X$. The distribution of the random variable comes from a parameterized family of distributions $\{p(x|\theta)\}_{\theta\in\Theta} \subseteq \Delta_\X$, where $\Theta$ is the parameter space.\footnote{By convention $p(x|...)$ often refers to the entire distribution, rather than the density value at a particular $x$; the usage should be clear from context.} There exists a prior distribution $p(\theta)$ over the parameters. Both $\{p(x|\theta)\}_{\theta\in\Theta}$ and $p(\theta)$ are common knowledge to the agents and the principal. 

Nature draws the true parameter $\theta^*$, which is unknown to both the agents and the principal, according to the prior $p(\theta)$. Each agent then receives some number of samples from $\X$ which drawn independently according to $p(x|\theta^*)$. In other words, if $x_1,\ldots,x_N$ is an enumeration of all samples received by any of the agents, then $p(x_i,x_j|\theta^*) = p(x_i|\theta^*)p(x_j|\theta^*)$ for all $i,j$ and all $\theta^* \in \Theta$. 


Agents form their beliefs about the random variable according to the Bayes' rule. If an agent receives samples $x_1,\ldots,x_{N}$, then we write the agent's belief as
\begin{multline}
  \label{eq:conf-ppd}
  p = p(x|x_1,\ldots,x_{N}) = \int_\Theta p(x|\theta) p(\theta|x_1,\ldots,x_{N}) d\theta \\\propto \int_\Theta p(\theta) p(x|\theta) \prod_j p(x_j|\theta) d\theta ~.
\end{multline}
This distribution is known as the \emph{posterior predictive distribution (PPD)} of $x$ given samples $x_1,\ldots,x_{N}$, and will be a central object of our analysis.

\subsection{Elicitation and Scoring Rules}
\label{sec:conf-scoring-rules}

An important feature of our model is that the principal has access to a sample $x\in\X$ herself, and can leverage this sample 
using scoring rule techniques
to elicit information from the agents. The principal's sample is also independently drawn according to $p(x|\theta^*)$. 
(In Section~\ref{sec:conf-non-unique-case}, we will allow the principal to have two such samples.)

The principal will choose a report space $\R$ and a scoring mechanism $S:\R\times\X\to\reals$, and request a report $r_i\in\R$ from each agent $i$. Upon receiving her sample $x$, the principal will give each agent a score of $S(r_i,x)$.  We assume that agents seek to maximize their expected score, so that if agent $i$ believes $x\sim p$ for some $p\in\Delta_\X$, then he will report $r_i \in \argmax_{r\in\R} \E_{x\sim p} [S(r,x)]$.  

Strictly proper scoring rules~\cite{Brier:50,gneiting2007strictly} are the basic tools for designing such scores $S$ that provide good incentive properties. 
A scoring rule is strictly proper if and only if reporting one's true prediction uniquely maximizes the expected score. Strictly proper scoring rules are most commonly used for eliciting a distribution over a finite outcome space, but also extend naturally to eliciting distributions with continuous support ~\cite{Matheson:76} and properties of distributions such as moments~\cite{gneiting2007strictly}. For example, the logarithmic scoring rule 
\begin{equation}
  \label{eq:conf-log-score}
  S(p,x) = \log p(x)
\end{equation}
is a popular strictly proper scoring rule for eliciting a distribution over a finite $\X$, where $p(x)$ is the reported probability for outcome $x$. Another popular strictly proper scoring rule, the Brier score~\cite{Brier:50}, can be used to elicit the mean of a random variable $\E [x]$, when taking the following form
\begin{equation}
  \label{eq:brier-mean}
  S(r,x) = 2r x - r^2
\end{equation}
or the first $k$ moments $(\E [x], \dots \E [x^k])$, when used as 
\begin{equation}
  \label{eq:brier-moments}
S(r_1,\ldots,r_k,x) = \sum_{i=1}^k 2r_i x^i - r_i^2. 
\end{equation}

\subsection{Aggregation}

The goal of the principal is to aggregate the information of the agents to obtain an accurate distribution of the random variable as if she has  
access to all of the samples from all agents. Throughout the paper, we will denote by $X$ this multiset\footnote{We use multisets, or equivalently unordered lists, as when $\X$ is a finite set it is likely that samples will not be unique.} of all observed samples by agents.

\begin{definition}
  \label{def:conf-opt-agg}
  Given prior $p(\theta)$ and data $X$ distributed among the agents, the \emph{global posterior predictive distribution (global PPD)} is the posterior predictive distribution $p(x|X)$.
\end{definition}

The goal of this paper is to design mechanisms which truthfully elicit information from agents in such a way that the global PPD $p(x|X)$ can be computed.  We capture this desideratum in the following definition.

\begin{definition}
  \label{def:conf-mech-opt-agg}
  Let $S : \R\times \X \to \reals$ be given, and let each agent $i$ receive samples $X^i$, with $X = \uplus_i X^i$ (multiset addition).  Let $r_i$ be the report of agent $i$, namely $r_i = \argmax_r \E_{p(x|X^i)}[S(r,x)]$.  Then $S$ \emph{achieves optimal aggregation} if there exists some function $g:\R^m \to \Delta_\X$ such that $g(r_1,\cdots,r_m) = p(x|X)$.
\end{definition}

It is worth noting that the report space $\R$ of the elicitation mechanism is often different from the space of PPD, i.e. $\Delta_\X$. In fact, we will design elicitation mechanisms such that the elicited reports help the principal to infer the {\em confidence} of agents, capturing the amount of samples that the agents have experienced, which then enables the optimal aggregation. This leads to our focus on the conjugate prior family. 

As a motivating example, consider the Normal distribution case, with $p(x|\theta) = \N(\theta,1)$ and $p(\theta) = \N(\mu,1)$, where $\N(\mu,\sigma^2)$ is the normal distribution with mean $\mu$ and variance $\sigma^2$.  It is well known that an agent $i$ has posterior distribution $p(\theta| X_i) = \N\bigl((\mu + x_1 + \cdots + x_{n_i})/(n_i+1),1/(n_i+1)\bigr)$ after observing samples $X_i = \{x_1,\ldots,x_{n_i}\}$. His estimate of the mean of $\theta$ is the weighted sum of his sample and the prior mean. The inverse of the variance, $n_i+1$, is called the {\em precision}, which encodes the agent's confidence or experience. Hence, if the principal can elicit mean estimate $\mu_i$ and precision $n_i +1$ from each of the $m$ agents, he can calculate the global PPD, which is a Normal distribution with mean $\tfrac 1 {N+1} \left(\mu + \sum_i n_i \mu_i\right)$ and variance $\tfrac 1 {N+1}$, where $N = \sum_i n_i$. This is the case studied by \citeauthor{fang2007putting} (\citeyear{fang2007putting}). We will see next that the general notion of conjugate priors will allow us to preserve the important aggregation properties we require elegantly.


\subsection{Conjugate Priors}

In this paper, we focus on prior distributions $p(\theta)$ that come from the conjugate prior family for distributions $\{p(x|\theta)\}_{\theta\in \Theta}$. This ensures that the posterior distribution on $\theta$ is in the same family of distributions as the prior $p(\theta)$ and also simplifies the optimal aggregation problem.    


While many notions of conjugate priors appear in the literature~\cite{fink1997compendium,gelman2013bayesian}, we adopt the following definition, which says that the conjugate prior family is parameterized by \emph{hyperparametrs} $\nu$ and $n$ which are {\em linearly} updated after observing samples: the new parameters can be written as a linear combination of the old parameters and sufficient statistics for the samples.

\begin{definition}
  \label{def:conf-conj-prior}
  Let $P=\{p(x|\theta):\theta\in\Theta\} \subseteq \Delta_\X$ be given.  A family of distributions $\{p(\theta|\nu,n) : \nu\in\reals^k, n\in\reals_+\} \subseteq \Delta_\Theta$ is a \emph{conjugate prior} family for $P$ if there exists a statistic $\phi:\X\to\reals^k$ such that, given the prior distribution $p(\theta|\nu_0, n_0)$, the posterior distribution on $\theta$ after observing $x$,
  \begin{equation}
    \label{eq:conf-post-update}
    p(\theta|\nu_0,n_0,x)
    = \frac {p(\theta|\nu_0,n_0) p(x|\theta)} {\int_\Theta p(\theta'|\nu_0,n_0) p(x|\theta') d\theta'},
  \end{equation}
  is equal to $p(\theta|\nu_0+\phi(x),n_0+1)$ for all  $\nu_0$ and $n_0$. 
\end{definition}

Using conjugate priors, the optimal aggregation problem simplifies considerably.  Given prior $p(\theta|\nu_0,n_0)$ and data $X=\{x_1,\ldots,x_N\}$ distributed among the agents, the global PPD can be written succinctly as
  \begin{equation}
    \label{eq:conf-opt-agg}
p(x|\nu_0,n_0,X) = p\left(x\;\Big|\;\nu_0+\textstyle\sum\limits_{i=1}^{N}\phi(x_i),\,n_0+N\right).
  \end{equation}
We can see that as we require $n$ to update by $1$ for each additional sample, $n - n_0$ exactly corresponds to the number of samples seen in total.  This is precisely the notion of confidence we wish to quantify --- the amount of data or experience that led to a prediction.  In particular, if we could obtain the hyperparameters $(\nu_i,n_i)$ for an agent's report, we could directly compute the number of samples $N_i = n_i - n_0$ they observed, as well as the sum of the sufficient statistics of their samples, $\sum_{x\in X^i} \phi(x)$.
If the principal can gather these two quantities from each agent $i$, then using the identities
$\sum_{x\in X^i} \phi(x) = \nu_i - \nu_0$ and $N_i = n_i - n_0$, the principal can aggregate these parameters by the observation that
\begin{align}
  \sum_{i=1}^N \phi(x_i) &= \sum_{i=1}^m \sum_{x\in X^i} \phi(x) = \sum_{i=1}^m (\nu_i - \nu_0) \label{eq:conf-agg-nu}\\
  N &= \sum_{i=1}^m N_i = \sum_{i=1}^m (n_i - n_0)~.
  \label{eq:conf-agg-n}
\end{align}
From here, the principal simply plugs these values into eq.~\eqref{eq:conf-opt-agg} to obtain the global PPD.

\section{Unique Predictive Distributions}
\label{sec:conf-uniq-pred-distr}

In this section, we show how the principal can leverage a single sample $x\in\X$ to elicit the hyperparameters of the posterior distributions of the agents, provided that the mapping from hyperparameters to predictive posterior distributions is unique.  Note that this statement contains two different types of posterior distributions, and as the distinction is important we take a moment to recall their differences.  After making his observations, an agent will have updated his hyperparameters to $(\nu,n)$.   This gives him a {\em posterior} distribution $p(\theta| \nu,n)$ over the parameter of the random variable and a {\em predictive posterior} distribution (PPD) $p(x| \nu,n)$ of the random variable itself.  

We begin with two simple but important results.  The first is an analog of the revelation principle from economic theory, showing that the most a principal with a single sample $x\in\X$ can get from an agent is the agent's private belief $p\in\Delta_\X$ about $x$.

\begin{lemma}
  \label{lem:conf-folk}
  Given a sample $x\in\X$ which an agent believes to be drawn from $p\in\Delta_\X$, any information obtained with a mechanism $S:\R\times\X\to\reals$, from an agent maximizing his expected score, can be written as a function of $p$.
\end{lemma}

\begin{proof}
  We need only find a function $f:\Delta_\X\to\R$ such that $f(p) \in \argmax_{r\in\R} \E_{x\sim p}[S(r,x)]$ whenever the $\argmax$ exists.
  Let $r_0\in\R$ be arbitrary.  For all $p\in\Delta_\X$, simply select $r_p \in \argmax_{r\in\R} \E_{x\sim p}[S(r,x)]$, or $r_p = r_0$ if the $\argmax$ is not defined, and let $f(p) = r_p$.
\end{proof}

While intuitive and almost obvious, Lemma~\ref{lem:conf-folk} is quite useful when thinking about elicitation problems.  For example, in our setting it is certainly clear that the principal can take $\R = \Delta_\X$ and use any strictly proper scoring rule to get the agent's PPD $p(x|\nu,n)$.
One might be tempted, however, to try to get more information: if one could simply elicit the posterior $p(\theta|\nu,n)$, then the hyperparametrs $(\nu,n)$ would be readily available for aggregation.  One tantalizing scheme would be to compute the distribution $p(\theta|x)$ 
and draw a sample $\hat\theta\sim p(\theta|x)$,
and then use this $\hat\theta$ to elicit $p(\theta|\nu,n)$ using the log scoring rule~\eqref{eq:conf-log-score}.  Lemma~\ref{lem:conf-folk} says that, while this may succeed, it will only succeed when the principal could have simply computed $p(\theta|\nu,n)$ from the PPD $p(x|\nu,n)$ to begin with.

For precisely this reason, we will see that being able to map the PPD to the posterior distribution is crucial to being able to optimally aggregate.  Before proving this, we need to introduce some more precise notation to describe the relationship between the hyperparameters and the PPD.

\begin{definition}
  \label{def:conf-reachable}
  Given hyperparameters $(\nu_0,n_0)$, we say $(\nu,n)$ is \emph{reachable from $(\nu_0,n_0)$} if there exists a multiset $X$ of $\X$ such that $\nu = \nu_0 + \sum_{x\in\X} \phi(x)$ and $n = n_0 + |X|$.  Additionally, we define the relation $(\nu,n)\equiv(\nu',n')$ if for all such $X$, including $\emptyset$, we have $p(x|\nu,n,X) = p(x|\nu',n',X)$.
\end{definition}

\begin{theorem}
  \label{thm:conf-unique-pred}
  Given a family of distributions $\{p(x|\theta)\}$ and conjugate prior $p(\theta|\nu_0,n_0)$, there exists a mechanism $S$ achieving optimal aggregation if and only if for all $(\nu,n)$ and $(\nu',n')$ reachable from $(\nu_0,n_0)$ we have that $p(x|\nu,n) = p(x|\nu',n')$ implies $(\nu,n)\equiv(\nu',n')$.
\end{theorem}

\begin{proof}
  We first prove the if direction.
  Let $S$ be the log scoring rule~\eqref{eq:conf-log-score}; then by propriety, the principal elicits
  $p_i = p(x|\nu_0,n_0,X_i) = p(x|\nu_i,n_i)$ for all $i$.  From $p_i$ the principal cannot necessarily compute $(\nu_i,n_i)$, but she can choose some $(\nu_i',n_i')$ reachable from $(\nu_0,n_0)$ such that $p_i = p(x|\nu_i',n_i')$.  We will show that since $(\nu_i,n_i)\equiv(\nu_i',n_i')$, this is enough to optimally aggregate.  We will restrict to the case of two agents; the rest then follows by induction.  Let $\phi(X) = \sum_{x\in X} \phi(x)$; by reachability, we have $X_1'$, $X_2'$ such that $\nu_i' = \nu_0 + \phi(X_i')$ and $n' = n_0 + |X_i'|$.  Thus,
  \begin{align*}
    &  p(x | \nu_0 + \textstyle\sum_i(\nu_i'-\nu_0), n_0 + \textstyle\sum_i(n_i'-n_0) )
    \\ &= p(x | \nu_2' + (\nu_1'-\nu_0), n_2' + (n_1'-n_0) )
    \\ &= p(x | \nu_2' + \phi(X_1'), n_2' + |X_1'| )
    \\ &\overset{*}{=} p(x | \nu_2 + \phi(X_1'), n_2 + |X_1'| )
    \\ &= p(x | \nu_2 + (\nu_1'-\nu_0), n_2 + (n_1'-n_0) )
    \\ &= p(x | \nu_1' + (\nu_2-\nu_0), n_1' + (n_2-n_0) )
    \\ &= p(x | \nu_1' + \phi(X_2), n_1' + |X_2| )
    \\ &\overset{*}{=} p(x | \nu_1 + \phi(X_2), n_1 + |X_2| )
    \\ &= p(x | \nu_1 + (\nu_2-\nu_0), n_1 + (n_2-n_0) )
\\ &= p(x | \nu_0 + \textstyle\sum_i(\nu_i-\nu_0), n_0 + \textstyle\sum_i(n_i-n_0) )~,
  \end{align*}
  which is the global PPD.  The starred equations used the fact that $(\nu_i,n_i)\equiv(\nu_i',n_i')$.

  For the only-if direction, assume that there are 
  $X,X'$ such that for $\nu = \nu_0 + \phi(X)$ and $\nu' = \nu_0 + \phi(X')$, we have $p(x|\nu,n) = p(x|\nu',n')$ but $(\nu,n)\not\equiv(\nu',n')$.  Then we have some multiset $X_1$ of $\X$ such that $p(x|\nu,n,X_1) \neq p(x|\nu',n',X_1)$.  Now let agent 1 receive $X_1$, and consider two worlds, one in which $X_2 = X$ and the other in which $X_2 = X'$.  By Lemma~\ref{lem:conf-folk}, without loss of generality, the principal uses $S$ to elicit the PPD from both agents.  However, she cannot distinguish between these two worlds, as by assumption agent 2's PPD is the same in both.  Unfortunately, the global PPDs in these two situations are different:
  \begin{align*}
    p(x | \nu_0,n_0,X_1\uplus X) &= p(x | \nu,n,X_1)
    \\ &\neq p(x | \nu',n',X_1)
    \\ &= p(x | \nu_0,n_0,X_1\uplus X')~.
  \end{align*}
Hence, the principal is unable to optimally aggregate.
\end{proof}

An important corollary of Theorem~\ref{thm:conf-unique-pred}, which we will make extensive use of below, is that the principal can always optimally aggregate if the PPD gives her full information about the hyperparameters.
\begin{corollary}
  \label{cor:conf-unique}
  If the map $\varphi:(\nu,n)\mapsto p(x|\nu,n)$ is injective, the principal can optimally aggregate.
\end{corollary}

\begin{proof}
  By injectivity, $p(x|\nu,n) = p(x|\nu',n')$ implies $(\nu,n) = (\nu',n')$, and $\equiv$ is an equivalence relation.  Moreover, any strictly proper scoring rule $S$ suffices as the mechanism, as this will elicit the PPD $p$, and then the principal can compute $(\nu,n) = \varphi^{-1}(p)$.
\end{proof}

In the following, we provide several examples illustrating the utility of Theorem~\ref{thm:conf-unique-pred}, and Corollary~\ref{cor:conf-unique} in particular.  Before continuing, however, we would like to remark on some practical consideratons.  Strictly speaking, the mechanism given by Corollary~\ref{cor:conf-unique}, which elicits the PPD and inverts the map $\varphi$, suffices when the modeling assumptions are all correct.  However, in the case where the model is slightly off, be it in our conditional independence assumption, the core family $p(x|\theta)$, or even the particular choice of prior, this approach appears to provide no guarantees.  In the examples that follow, we seek not only to elicit the hyperparameters of the PPD, but to do so using scoring rules which provide meaningful information about the PPD \emph{regardless of its form}.  For example, we show below how to elicit the PPD for the Poisson distribution with a Gamma prior using a scoring rule for the first and second moment (or equivalently, the mean and variance).  This scoring rule has the property that it will elicit the correct moments of \emph{any} distribution, and thus if the agents' PPD does not have the assumed form, a practitioner would still have meaningful information about the agent's belief for a variety of approximate aggregation techniques.

\vspace{-5pt}\paragraph{Poisson}

Imagine that a citizen science project such as eBird~\cite{sullivan2009ebird} wishes to collect observations about sightings of various birds to deduce bird migration patterns.  Such a project may wish users to report the number of birds of a particular species seen per minute.  Of course, to combine such estimates, eBird would like to know not only the observed rate, but how long the user spend bird watching, so that it may weigh more highly reports from longer time intervals; this is precisely what our approach offers.

For situations such as this one which involve counting events in a specified time interval, the Poisson distribution is a common choice.  The parameter of the Poisson distribution is $\lambda\in\reals$, the \emph{rate} parameter, and the probability of observing $x\in\{0,1,2,\ldots\}$ events in a unit time interval is given by $p(x|\lambda) = \lambda^x e^{-\lambda} / x!$.  
The canonical conjugate prior for the Poisson distribution is the Gamma distribution, given by $p(\lambda|\nu,n) = \frac{n^{\nu}}{\Gamma(\nu)} \lambda^{\nu-1} e^{-n\,\lambda}$, and the statistic is $\phi(x) = x$.  The form of the PPD $p(x|\nu,n)$ is also a familiar distribution, in the negative binomial family~\cite[p.44]{gelman2013bayesian}.

As mentioned above, we will show how to compute the hyperparameters $\nu$ and $n$ of the PPD from its first two moments $\mu_1$ and $\mu_2$.  As the form of the PPD is known to be negative binomial, one can easily calculate or look up what these moments are in terms of the hyperparameters: $\mu_1 = \nu/n$ and $\mu_2 = \nu (\nu + n + 1)/n^2$.  Fortunately, given these equations, we can simply solve for the hyperparameters in terms of the moments, which we can elicit robustly: $n = \mu_1/(\mu_2+\mu_1^2+\mu_1)$ and $\nu = n \mu_1$.  This already verifies the injectivity condition of Corollary~\ref{cor:conf-unique}, so we know that optimal aggregation is possible.

For concreteness, let us return to the bird watching example to show how eBird might reward users in such a way as to truthfully obtain predictions and then compute their optimal aggregation.  The protocol would be for eBird to announce that a representative will be sent tomorrow to count the number $x$ of birds seen in a minute, and to ask each user $i$ for a prediction $r_{i,1}$ about $\E [x]$ and $r_{i,2}$ about $\E[x^2]$, with the understanding that after the count $x$ is revealed, agent $i$ will receive a reward (cf.~\eqref{eq:brier-moments}) of
\begin{equation}
  \label{eq:conf-poisson-score}
  S(r_{i,1},r_{i,2},x) = 2r_{i,1}x - r_{i,1}^2 + 2r_{i,2}x^2 - r_{i,2}^2~.
\end{equation}
With the reports in hand, eBird can compute $n_i = r_{i,1}/(r_{i,2}+r_{i,1}^2+r_{i,1})$ and $\nu_i = n_i r_{i,1}$.  Assuming the common prior parameters $(\nu_0,n_0)$ are known, eBird simply aggregates these reports to $n = n_0 + \sum_{i=1}^m (n_i - n_0)$ and $\nu = \nu_0 + \sum_{i=1}^m (\nu_i - \nu_0)$, arriving at the global PPD $p(x|X) = p(x|\nu,n)$.

\vspace{-5pt}\paragraph{Normal}  As we saw in Section~\ref{sec:conf-model}, the Normal distribution with known variance but unknown mean allows for optimal aggregation.  This follows easily from Corollary~\ref{cor:conf-unique} as well, since $\N(\mu,\sigma^2)$ is a different distribution for each setting of $\mu,\sigma$.

\vspace{-5pt}\paragraph{Uniform}
\raf{Later: cite Bernardo and Smith?}

Perhaps the most natural of distributions is the uniform distribution on $[0,\theta]$, where $p(x|\theta) = 1/\theta$ in that interval.  As a simple application, consider the problem of determining the number of raffle tickets sold at a fair by asking random people what their ticket number is.  It is well-known that the Pareto distribution is a conjugate prior for this case, and the hyperparameter update is $\nu = \max(\nu_0,x)$ and $n = n_0 + 1$.  Observe that the hyperparameter update is not linear, so we cannot simply apply Corollary~\ref{cor:conf-unique}.  However, it is easy to see that the conclusion still holds here, as the principal can easily aggregate $\{(\nu_i,n_i)\}_i^m$ by taking $\nu = \max\{\nu_i\}_{i=0}^m$ and $n = n_0 + \sum_i(n_i-n_0)$ as usual.

By a simple calculation, one can show that the PPD in this case is a mixture of a uniform distribution and a Pareto distribution, from which one can compute the moments $\mu_1 = n\nu/2(n-1)$ and $\mu_2 = n\nu^2/3(n-2)$.  Cancelling $\nu$, these equations give a quadratic equation with a unique root $n$ satisfying $n>2$ (a requirement of the prior), from which $\nu$ can also be calculated.  Thus, the principal can achieve optimal aggregation in this case as well.


\section{The Non-Unique Case}
\label{sec:conf-non-unique-case}

Imagine a setting where the principal wants to aggregate information from agents to estimate the bias of a coin.  The principal asks agents Bob and Carol, who each see some unknown number of coin flips, after which Bob reports that the coin is unbiased, whereas Carol reports that it is biased 10-to-1 toward Heads.  With only this information, which corresponds to the full PPDs of both agents, it is easily seen to be \emph{impossible} to optimally aggregate these reports, as it is unclear how many flips each agent saw.  Even if the principal knows that Carol saw 20 flips, she cannot tell whether Bob saw none and just reported the prior, or whether he saw 1000 and is practically certain of the bias of the coin.  (Formally, we can explain this by noting that the conjugate prior is the Beta distribution, which does not satisfy Theorem~\ref{thm:conf-unique-pred}.)  How can the principal circumvent this impossibility to still achieve optimal aggregation in this setting?

In this section we will consider a more general version of the coin flip example, using the \emph{categorical} family of distributions, i.e., the whole of $\Delta_\X$ for $\X = [K] = \{1,2,\ldots,K\}$.  Here the common conjugate prior is the Dirichlet distribution $p(\theta|\alpha)$, whose hyperparameters $\alpha \in \reals^K$ encode \emph{pseudo-counts}, so that $\alpha_i$ corresponds to the number of occurrences of outcome $i$ an agent has seen.  More formally, we take $\Theta = \Delta_\X = \Delta_K$, and for $\alpha \in \reals^K$ we let
\begin{multline}
  \label{eq:conf-dirichlet}
  p(i|\theta) = \theta_i\,,\quad
  p(\theta|\alpha) = \frac{\Gamma(n)}{\prod_{i=1}^K \Gamma(\alpha_i)} \prod_{i=1}^K \theta_i^{\alpha_i - 1}~,
\end{multline}
where $n = \sum_{i=1}^K \alpha_i$ corresponds to the total number of (pseudo-) samples observed, and $\Gamma$ is the Gamma distribution.\footnote{Note that we have departed from our $(\nu,n)$ notation to match the convention for the Dirichlet distribution; otherwise we could take $\nu$ to be the first $K-1$ coordinates of $\alpha$, and keep $n$ the same.} It is well-known that the mean of the Dirichlet distribution  is $\E[\theta|\alpha] = \alpha/n$, which is just a normalized version of the pseudo-counts.  Taken as an element of $\Delta_\X$, this is also the PPD: if an agent sees $x\!=\!1$ and $x\!=\!2$ each eight times and $x\!=\!3$ four times, then $\alpha=(8,8,4)$ and his PPD will be $(2/5,2/5,1/5)$.
We can see now why Theorem~\ref{thm:conf-unique-pred} tells us that optimal aggregation is impossible: scaling $\alpha$ by any positive amount yields the same PPD, just as with the coin flip example above, but when aggregating $\alpha$'s from multiple agents, different relative scales yield different global PPDs.

Fortunately, despite this impossibility, we now show that if the principal can simply obtain \emph{two} of her own samples, she can use them both to glean second-order information from the agents, and then optimally aggregate.  The idea behind the mechanism is extremely simple: ask the agent for the distribution $p$ of the first sample, and the probability $b$ that the two samples are the same.  As discussed above, the reported $p$ gives $\alpha/n$, and it turns out that the scaling factor $n$, which corresponds to the confidence of the agent, can be expressed as a simple formula of $p$ and $b$.

\begin{theorem}
  \label{thm:conf-2-samples}
  Let $\X = [K]$, and let $\{p(i|\theta)\}$ and $\{p(\theta|\alpha)\}$ be the categorical and Dirichlet families from eq.~\eqref{eq:conf-dirichlet}.  Then given two independent samples $x_1,x_2\in\X$, the mechanism $S:\Delta_\X\times[0,1]\times\X\times\X\to\reals$ defined by
  \begin{equation}
    \label{eq:conf-1}
    S(p,b,x_1,x_2) = \log p(x_1) + 2 b\cdot\ones\{x_1\!=\!x_2\} - b^2
  \end{equation}
  achieves optimal aggregation.
\end{theorem}

\begin{proof}
  Focusing first on a single agent, by propriety of the log scoring rule, the agent will report $p = p(\,\cdot\,|\alpha) = \alpha/n$, where once again $n = \sum_{i=1}^K \alpha_i$.
  Similarly, by propriety of the Brier score, the agent will report his belief about the probability that $x_1=x_2$.  We can calculate this easily:
  \begin{align*}
    b &= \mathrm{Pr}[x_1=x_2]\\
    &= \E_{\theta\sim p(\theta|\alpha)}\left[\textstyle\sum\nolimits_{i=1}^K p(x_1=i,x_2=i\,|\theta)\right] \\
    &= \E_{\theta\sim p(\theta|\alpha)}\left[\textstyle\sum\nolimits_{i=1}^K p(x_1=i\,|\theta) p(x_2=i\,|\theta)\right] \\
    &= \E_{\theta\sim p(\theta|\alpha)}\left[\textstyle\sum\nolimits_{i=1}^K \theta_i \theta_i\right] 
    = \textstyle\sum\nolimits_{i=1}^K \Var[\theta_i|\alpha] + \E[\theta_i|\alpha]^2~.
  \end{align*}
  It is known that $\Var[\theta_i|\alpha] = \frac{\alpha_i(n - \alpha_i)}{n^2(n+1)}$, so the first term becomes
  \begin{equation*}
    \sum\nolimits_i \Var[\theta_i|\alpha] = \frac{(\sum_i \alpha_i)n - \sum_i\alpha_i^2}{n^2(n+1)} = \frac{1-\|p\|^2}{n+1},
  \end{equation*}
  as we also have $\sum_i \E[\theta_i|\alpha]^2 = \|p\|^2 = \|\alpha\|^2/n^2$.  Putting this together, we have $b = \frac{1-\|p\|^2}{n+1} - \|p\|^2$, so $n = \frac{1-b}{b-\|p\|^2}$ and finally $\alpha = n p$.  Finally, turning  to the aggregation of multiple predictions, the result follows by the same argument as in Theorem~\ref{thm:conf-unique-pred}: we simply discount the prior from each agent's report and sum.
\end{proof}

Returning to the coin flip example, we can now see how the principal can resolve the dilemma from before.  Instead of simply asking the probability that a single flip is Heads, the principal should obtain two independent flips and then ask the agents for the probability that the first is Heads, and the probability that the two flips are the same.  By Theorem~\ref{thm:conf-2-samples}, the answers to these two intuitive questions give the principal enough information to optimally aggregate.

\raf{Not sure it is worth going into this: For this specific binary outcome example, it is interesting to see how the aggregation works explicitly.  Here the distribution is Bernoulli and the prior is Beta...}
\ian{I think it would be worth it if we have space and time (or in the appendix if time but not space).  Otherwise this seems a bit like magic that everything works out.} \yc{I think it's worth going into this too.}
\raf{Decided not to; not enough time or space, plus it's really the same proof as above just with only two outcomes.}

\section{Future Work}
\label{sec:conf-conclusion}

A well known and broad class of distributions with conjugate priors are the exponential families (see Appendix~\ref{sec:conf-exponential-families} for a primer).  Many of the examples discussed in this paper fall into the exponential families, and thus it is a natural question to ask whether our results can be shown to hold for all such distributions.
In particular, our study opens two interesting questions, which under the surface would imply some interesting structure of exponential families.

The first follows naturally from Theorem~\ref{thm:conf-unique-pred} and the examples in Section~\ref{sec:conf-uniq-pred-distr}, several of which are exponential families, and all of which admit optimal aggregation.  We conjecture that for exponential families, the success of a single-sample mechanism depends only on the dimension $k$ of the statistic $\phi$.

\begin{conjecture}
  \label{con:conf-exp-fam-single}
  Optimal aggregation with a single sample is possible for an exponential family if and only if $|\X| > \dim \phi + 1$.
\end{conjecture}

The second open question is similar: does the two-sample technique from Section~\ref{sec:conf-non-unique-case} succeed for all exponential families?  Again, we conjecture positively.

\begin{conjecture}
  \label{con:conf-exp-fam-double}
  Given an exponential family with statistic $\phi$, the mechanism which elicits the expected values of $\phi(x_1)$ and $\phi(x_1)\phi(x_2)^\tr$ can optimally aggregate.
\end{conjecture}

The intuition behind these conjectures, which we outline in Appendix~\ref{sec:conf-conjectures}, lies in concentration properties in the posterior distribution $p(x|\nu,n)$ as $n$ increases to infinity.  Because of the simple form of exponential families, and the exponential decay inherent in their definition, we believe that these results can be obtained.

Finally, we would like to mention a possible extension.  While our model assumes that the principal wishes to aggregate \emph{all} information, in reality, agents may have different costs to gather their samples, and the principal may therefore desire to aggregate a more efficient amount of information given this cost.  \citeauthor{fang2007putting} (\citeyear{fang2007putting}) show that this can be done in a restricted setting with Normal distributions.  Can this still be done in our more general setting?  What if agents can acquire different amounts of information at different costs, for example, if a convex function specifies their cost to acquire any number of samples?  We hope to address these and related questions in future work.

\bibliographystyle{aaai}
\bibliography{../../diss,conf,library}

\raf{\subsection*{Acknowledgments}
\raf{Will flesh these out}
David (postdoc on MD 2nd floor), Giri, Matus, David Blei}
\ian{We can omit these from the submission since it is supposed to be double-blind.}

\newpage
\appendix
\section{Exponential families}
\label{sec:conf-exponential-families}

Perhaps the most important class of distributions which admit conjugate priors are the exponential families, a broad class which includes many common distributions such as normal, log-normal, Poisson, and many more.  We briefly review exponential families and their conjugate priors, which as it turns out are themselves exponential families.

Let $\phi:\X\to\reals^k$ be the \emph{sufficient statistic} (a term justified below).  We assume that $\phi$ is \emph{minimal}, meaning $\inprod{\theta,\phi(x)}$ cannot be a constant function of $x$ for any $\theta\neq 0$.  (Minimality is thus equivalent to affine independence.) Now define
\begin{align}
  g(\theta) &= \log\int_\X \exp\{\inprod{\phi(x),\theta}\} dx
  \\
  p(x|\theta) &= \exp\{\inprod{\phi(x),\theta} - g(\theta)\}.
\end{align}
This family $\{p(x|\theta)\}$ is the exponential familiy with respect to $\phi$.  We refer to $\Theta$ as the \emph{natural parameters}, as contrast to the \emph{mean parameters} $\mu(\theta) = \E[\phi|\theta]$, which also parameterize the family provided certain regularity conditions are met~\cite{wainwright2008graphical}.  The function $g$ is called the \emph{cumulant}, and happens to generate the moments of $\phi$ under $p(x|\theta)$.  In particular, we have $\nabla g(\theta) = \E[\phi|\theta] = \mu(\theta)$.

Turning now to the conjugate prior for this family, let 
\begin{align}
  h(\nu,n) &= \log\int_\Theta \exp\{\inprod{\theta,\nu} - n g(\theta)\} d\theta\\
  p(\theta|\nu,n) &= \exp\{\inprod{\theta,\nu} - n g(\theta) - h(\nu,n)\}.\label{eq:conf-exp-fam-conj}
\end{align}
One can verify directly that $p(\theta|\nu,n)$ is a conjugate family to $p(x|\theta)$.  Moreover, the priors are themselves exponential families, with respect to statistic
$\psi(\theta) = \left[\begin{smallmatrix}
      \theta\\-g(\theta)
    \end{smallmatrix}\right]$.

As we saw above, it is easy to verify by direct calculation that the cumulant $g(\theta)$ satisfies $\nabla g(\theta) = \E_{x\sim p(x|\theta)} [ \phi(x) ]$.  A much less obvious fact, but a very useful one, is that the implied mean $\nu/n$ of the conjugate prior $p(\theta|\nu,n)$ is \emph{credible}, in the sense that the expected value of $\phi$ is in fact $\nu/n$.

\begin{theorem}[\cite{diaconis1979conjugate}]
  \label{thm:conf-credible-mean}
  Let $p(x|\theta)$ be an exponential family with cumulant $g(\theta)$ and let $p(\theta|\nu,n)$ be its conjugate prior.  Then
  \begin{equation}
    \label{eq:conf-credible}
    \int_{\X} \phi(x) p(x|\nu,n) dx = \int_{\Theta} \nabla g(\theta) p(\theta|\nu,n) = \nu/n.
  \end{equation}
\end{theorem}

\raf{COMMENTED OUT: other examples}

\section{Conjectures for Exponential Families}
\label{sec:conf-conjectures}

Here we give intuition for the conjectures stated in Section~\ref{sec:conf-conclusion}.  For the first, Conjecture~\ref{con:conf-exp-fam-single}, note that when $\dim \phi = |\X|-1$, and the statistic is minimal, then $\Theta$ is just a reparameterization of the categorical distributions, $\Delta_\X$.  As we saw in Section~\ref{sec:conf-non-unique-case} that a single sample is insufficient for the categorical case, Conjecture~\ref{con:conf-exp-fam-single} is implied by the following alternate conjecture.

\begin{conjecture}
  \label{con:conf-exp-fam-injective}
  The map $\varphi:(\nu,n)\mapsto p(x|\nu,n)$ is injective for an exponential family conjugate prior if and only if $\dim \phi < |\X|-1$.
\end{conjecture}

There is considerable intuition for this conjecture.  By Theorem~\ref{thm:conf-credible-mean} (the credible mean property of exponential family conjugate priors), to examine the injectivity of $\varphi$ we may restrict our attention to a fixed valued of $\mu = \nu/n$.  This is because if $\nu/n\neq\nu'/n'$, then $\varphi(\nu,n) \neq \varphi(\nu',n')$.  Given this fact, it is clear that the injectivity cannot hold whenever $k \defeq \dim \phi \geq |\X|-1$, because by minimality of $\phi$, the mean $\E[\phi] = \nu/n = \mu$ must uniquely identify the distribution, and thus scaling $n$ and taking $\nu=n\mu$ and yields the same PPD for all $n>0$.  Conversely, one can show by the form of the conjugate prior~\eqref{eq:conf-exp-fam-conj} that for our fixed value of $\mu$, we have
\begin{equation}
  \label{eq:conf-2}
  \frac {p(\theta|n\mu,n)}{p(\theta'|n\mu,n)} = \left(\frac {p(\theta|\mu,1)}{p(\theta'|\mu,1)} \right)^n~,
\end{equation}
for all $\theta,\theta'\in\Theta$ and all $n,\mu$.  Thus, as $n$ increases there is strong concentration in the prior about the mode $\hat \theta$, which one can show is equal to $\nabla g^*(\nu/n)$ by convex conjugacy, so that $\mu(\hat\theta)=\mu$.  It is clear then that the limit of $p(x|n\mu,n)$ as $n\to\infty$ is simply $p(x|\hat\theta)$.  It would therefore be natural to show that $\mathrm{KL}(p(x|\nu,n) \,;\, p(x|\hat \theta))$, or some other notion of distance, is monotone decreasing in $n$, which would then imply injectivity of $\varphi$.

For Conjecture~\ref{con:conf-exp-fam-double}, the intuition lies in a reparameterization of the conjugate prior distribution.  Let $\mu(\theta) = \nabla g(\theta)$ denote the mean parameter corresponding to $\theta$, and recall from the credible mean property that $\E[\phi(x)|\nu,n] = \E[\mu(\theta)|\nu,n] = \nu/n$.  Then by independence of $x_1,x_2$, we have $\E[\phi(x_1)\phi(x_2)^\tr | \nu,n] = \E[\mu(\theta)\mu(\theta)^\tr | \nu,n]$.  Thus, letting $r_1$ and $R_2$ be the reported values for the $\E[\phi(x_1)]$ and $\E[\phi(x_1)\phi(x_2)^\tr]$, we see that $\Var[\mu(\theta)|\nu,n]$ is simply $R_2 - r_1r_1^\tr$.  In other words, we can use this information to compute the variance of the posterior distribution \emph{of the mean parameters}.  That is, if we thought of the posterior as being a distribution $p(\mu|\nu,n)$ over mean parameters instead of over natural parameters $\theta$, we would be able to elicit the variance of this posterior.  Intuitively, this variance should correspond to the confidence of the agent, and in particular should be monotone decreasing in $n$, which would allow us to compute $n$ and thus optimally aggregate.

\end{document}